%% file: main.tex
\documentclass[a4paper,conference]{IEEEtran}

\input{preamble}

\begin{document}

\title{Space-Time Codes Based on Rank-Metric Codes and Their Decoding}

\author{\IEEEauthorblockN{Sven Puchinger, Sebastian Stern, Martin Bossert, Robert F.H. Fischer\thanks{This work was supported by Deutsche Forschungsgemeinschaft (DFG) within the framework COIN under grants BO~867/29-3 and FI~982/4-3}}
\IEEEauthorblockA{
              Ulm University, 
              Institute of Communications Engineering, 
              89081 Ulm, Germany\\
              Email: \{sven.puchinger, sebastian.stern, martin.bossert, robert.fischer\}@uni-ulm.de}
}

\maketitle

\begin{abstract}
In this paper, a new class of space-time block codes is proposed. The new construction is based on finite-field rank-metric codes in combination with a rank-metric-preserving mapping to the set of Eisenstein integers.
It is shown that these codes achieve maximum diversity order and improve upon existing constructions.
Moreover, a new decoding algorithm for these codes is presented, utilizing the algebraic structure of the underlying finite-field rank-metric codes and employing lattice-reduction-aided equalization.
This decoder does not achieve the same performance as the classical maximum-likelihood decoding methods, but has polynomial complexity in the matrix dimension, making it usable for large field sizes and numbers of antennas.
\end{abstract}


% ===========================================================================
% Introduction
% ===========================================================================

\section{Introduction}

\noindent
Space-time (ST) codes were introduced in \cite{tarokh1998space} for multiple-input/multiple-output (MIMO) fading channels in point-to-point single-user (multi-antenna) scenarios.
Several code constructions have been proposed so far, both ST convolutional and block codes.
ST codes are usually maximum-likelihood (ML) decoded, yielding an exponential decoding complexity.

An important design criterion for ST codes is that the rank distance of two codewords must be as large as possible~\cite{tarokh1998space}.
In \cite{gabidulin2000space,lusina2003maximum}, finite-field rank-metric codes were used to construct ST block codes by mapping the finite-field elements to a modulation alphabet in the complex plane.
It was shown that this mapping preserves the minimum rank distance of the finite-field code in case of \emph{binary phase-shift keying} and subsets of the Gaussian integers $\GI$ \cite{bossert2002space}, as well as for other important constellations \cite{liu2002rank}.

In this paper, we prove that there is a rank-metric-preserving mapping in the case of Eisenstein integers $\EI$ \cite{conway2013sphere}.
The use of this modulation alphabet promises to improve upon other modulation alphabets in $\CC$, since Eisenstein integers form the hexagonal lattice in $\CC$, the densest possible lattice in a $2$-dimensional real vector space.

Furthermore, we present an alternative decoding method for these ST codes, using \emph{lattice-reduction-aided} (LRA) equalization techniques in combination with a decoding algorithm of the underlying finite-field rank-metric code.
This decoder is sub-optimal in terms of failure probability compared to the classical ML decoding methods, but has polynomial complexity and therefore can be used for a larger set of parameters.

The paper is organized as follows.
In Section~\ref{sec:preliminaries}, we describe the channel model and provide basics on Eisenstein integers and rank-metric codes.
We propose a new ST code construction in Section~\ref{sec:new_ST_construcion} and present alternative decoding methods in Section~\ref{sec:alternative_decoding}.
Section~\ref{sec:Conclusion} concludes the paper.

\section{Preliminaries}
\label{sec:preliminaries}

\subsection{Channel Model}
\label{subsec:channel_model}

We assume a flat-fading MIMO channel with additive white Gaussian noise, i.e.
\begin{align}
\Y = \H \X + \N,
\end{align}
and $\Ntx,\Nrx,\Ntime \in \NN$ denote the numbers of \emph{transmit antennas}, \emph{receive antennas} and \emph{time steps}, respectively, $\X \in \EI^{\Ntx \times \Ntime}$ is the \emph{sent codeword} and $\Y \in \CC^{\Nrx \times \Ntime}$ is the \emph{received word} (both over space (rows) and time (columns)).
$\H \in \CC^{\Nrx \times \Ntx}$ is the \emph{channel matrix}, which is known at the receiver (perfect channel state information) and whose entries are drawn i.i.d.\ from the zero-mean unit-variance complex Gaussian distribution.
Also, $\N \in \CC^{\Nrx \times \Ntime}$ is the \emph{noise matrix}, which is unknown at the receiver and whose entries are sampled i.i.d.\ from a zero-mean complex Gaussian distribution~\cite{tarokh1998space}.
The \emph{signal-to-noise} (SNR) ratio is given by the transmit energy per information bit $E_\mathrm{b,TX}$ in relation to the noise power spectral density $N_0$.

\subsection{Eisenstein Integers}

Let $\w = \eexp{\imag \frac{2 \pi}{3}}$.
Then the ring
\begin{align}
\EI := \ZZ[\w] = \{a+\w b : a,b \in \ZZ \} \subseteq \CC
\end{align}
is called \emph{Eisenstein integers} \cite{conway2013sphere}.
$\EI$ is a \emph{principal ideal domain} (PID), a \emph{Euclidean domain}, and a lattice.
The units of $\EI$ are the sixth roots of unity $\EIunits = \{\eexp{\imag \frac{\ell \pi}{6}} : \ell=1,\dots,6\}$.
Let $\Eprime \in \EI \setminus \{0\}$. Then $\Eprime \EI$ is a sub-lattice of $\EI$ and for any $z \in \CC$ we can define a quantization function
\begin{align}
\Equant{\Eprime \EI}{z} = \argmin\limits_{y \in \Eprime \EI} |z-y|
\end{align}
and a modulo function
\begin{align}
\Emod{\Eprime \EI}{z} = z - \Equant{\Eprime \EI}{z}.
\end{align}
Both $\Equant{\Eprime \EI}{\cdot}$ and $\Emod{\Eprime \EI}{\cdot}$ can be extended to vector or matrix inputs by applying them component-wise.
The \emph{Eisenstein integer constellation} of $\Eprime \in \EI \setminus \{0\}$ is the set
\begin{align}
\EC{\Eprime} = \{\Emod{\Eprime \EI}{z} : z \in \EI\}.
\end{align}
Note that $\EC{\Eprime} = \Voronoi{\Eprime \EI} \cap \EI$, where $\Voronoi{\Eprime \EI}$ is the \emph{Voronoi region} of the lattice $\Eprime \EI$ \cite{fischer2005precoding,stern2015lattice}.
$\EC{\Eprime}$ contains $|\Eprime|^2$ elements.
The resulting signal constellation has a hexagonal boundary region and is more densely packed than a signal constellation of the same cardinality over the \emph{Gaussian integers} or \emph{quadrature amplitude modulation}, cf.~\cite{stern2015lattice}.

Besides its high packing density, Eisenstein integer constellations have another major advantage compared to classical signal constellations: they possess algebraic structure.
In order to use this fact, we need the following lemma.
\begin{lemma}[\!\!\cite{conway2013sphere}]
$\Eprime$ is a prime in $\EI$ if one of the following conditions is true.
\begin{description}
\item[(i)] $\Eprime = u \cdot p$ for some $u \in \EIunits$ and $p$ is a prime in $\NN$ with $p \equiv 2 \mod 3$ (Type $\mathrm{I}$).
\item[(ii)] $|\Eprime|^2 = p$ is a prime in $\NN$ with $p \equiv 1 \mod 3$ or $p=3$ (Type $\mathrm{II}$).
\end{description}
\end{lemma}
We define multiplication and addition of $a,b \in \EC{\Eprime}$ as
\begin{align}
a \Eadd b = \Emod{\Theta \EI}{a + b} \quad \text{and} \quad
a \Emul b = \Emod{\Theta \EI}{a \cdot b},
\end{align}
where $+$ and $\cdot$ are the ordinary operations in $\CC$.
Then the set $\EC{\Eprime}$ with these operations $(\EC{\Eprime},\Eadd,\Emul)$ is a ring and---even stronger---the following theorem holds.
\begin{theorem}[\!\!\cite{conway2013sphere}]\label{thm:isomorphism}
Let $\Eprime$ be a prime in $\EI$. Then $(\EC{\Eprime},\Eadd,\Emul)$ is a finite field. More precisely, the following isomorphisms hold.
\begin{align}
(\EC{\Eprime},\Eadd,\Emul) \cong
\begin{cases}
\F_{p^2}, &\text{if $\Eprime$ is of Type $\mathrm{I}$}, \\
\F_{p}, &\text{if $\Eprime$ is of Type $\mathrm{II}$}.
\end{cases}
\end{align}
\end{theorem}
A table of suitable Eisenstein integer constellations of size up to $127$ can be found in \cite[Table~I]{stern2015lattice}, along with a list of constellations which are subsets of the Gaussian integers.
The table also states their resulting average power (mean squared absolute value of a constellation).

\subsection{Rank-Metric and Gabidulin Codes}
\label{subsec:Gabidulin}

Rank-metric codes are sets of matrices where the distance of two elements is measured by the \emph{rank metric} instead of the classical \emph{Hamming metric}.
The most famous class of rank-metric codes are Gabidulin codes, which were independently introduced in \cite{Delsarte_1978,Gabidulin_TheoryOfCodes_1985,Roth_RankCodes_1991} and are used in many applications such as \emph{random linear network coding} \cite{koetter2008coding} and cryptography \cite{gabidulin1991ideals}.

In general, a \emph{rank-metric code} $\Ccode$ over a field $\KK$ is a subset of $\KK^{m \times n}$, along with the \emph{rank metric}
\begin{align}
\dR : \KK^{m \times n} \times \KK^{m \times n} &\to \{0,\dots,\min\{n,m\}\}, \\
(\A,\B) &\mapsto \rk(\A-\B).
\end{align} It has \emph{minimum rank distance}
\begin{align}
d := \min\limits_{\substack{\C_1,\C_2 \in \Ccode\\ \C_1 \neq \C_2}} \dR(\C_1,\C_2).
\end{align}

Let $q$ be a prime power and $m \in \NN$.
Thus, $\Fqm$ can be seen as a vector space of dimension $m$ over $\Fq$ and for some $n \in \NN$, there is a mapping
\begin{align}
\ext : \Fqm^n \to \Fq^{m \times n}, \; \c \mapsto \C,
\end{align}
where each component of the vector $\c$ is extended into a fixed basis\footnote{E.g., $\mathcal{B}=\left(1,\alpha, \alpha^2, \dots, \alpha^{m-1}\right)$, where $\alpha$ is a primitive element of $\Fqm$.} $\mathcal{B}$ of $\Fqm$ over $\Fq$.
The expansion of the $i$th component of $\c$ is then the $i$th column of $\C$.
A \emph{linearized polynomial} over $\Fqm$ of $q$-degree $d_f \in \NN_0$ is a polynomial of the form
\begin{align}
f(X) = \sum\limits_{i=0}^{d_f} f_i X^{q^i}, \quad f_i \in \Fqm, \quad f_{d_f} \neq 0.
\end{align}
The zero polynomial $f(X)=0$ is also a linearized polynomial and has $q$-degree $d_f = -\infty$.
The set of linearized polynomials over $\Fqm$ is denoted by $\Lset$.

Let $k,n \in \NN$ be such that $k < n \leq m$.
We choose $g_1,\dots,g_n \in \Fqm$ to be linearly independent over $\Fq$.
A \emph{Gabidulin code} of length $n$ and dimension $k$ is given by
\begin{align}
\CGab[n,k] = \{[f(g_1),\dots,f(g_n)] : f \in \Lset, d_f < k\}.
\end{align}
The codewords $\c = [f(g_1),\dots,f(g_n)] \in \Fqm^n$ can be interpreted as matrices $\C \in \Fq^{m \times n}$ using the $\ext$ mapping and thus, the rank metric is well-defined.
The minimum rank distance of $\CGab[n,k]$ is $d = n-k+1$ and therefore fulfills the rank-metric Singleton bound with equality \cite{Delsarte_1978,Gabidulin_TheoryOfCodes_1985,Roth_RankCodes_1991}.

It is shown in \cite{silva2008rank}\footnote{In the published version of this paper, a wrong reference was given.} that we can reconstruct $\C \in \CGab$ from
\begin{align}
\C + \E + \AR \BR + \AC \BC,
\end{align}
where $\rk(\E) = \tau$, $\AR \in \Fq^{m \times \rho}, \BR \in \Fq^{\rho \times n}$, $\AC \in \Fq^{m \times \delta}, \BC \in \Fq^{\delta \times n}$,
whenever
\begin{align}
2 \tau + \rho + \delta < d \label{eq:dec_cond}
\end{align}
and $\AR$ and $\BC$ are known at the receiver (\emph{error and erasure decoder}).
The decoding complexity is $\BigO(m^3)$ operations in $\Fq$ \cite{silva2008rank}, or $\Otilde(n^{1.69} m)$ using the algorithms in \cite{puchinger2016subquadratic}, where $\Otilde$ is the asymptotic complexity neglecting $\log(nm)$ factors.

A \emph{criss-cross} error is a matrix that contains non-zero entries only in a limited number of rows and columns, cf.~\cite{Roth_RankCodes_1991}.
In general, if such a matrix can be covered with $\tau_\mathrm{r}$ rows and $\tau_\mathrm{c}$ columns such that outside the cover, there is no error, the matrix has rank $\leq \tau_\mathrm{r} + \tau_\mathrm{c}$.
Therefore, criss-cross and rank errors are closely related.

\section{A New Construction Based on Eisenstein Integers}
\label{sec:new_ST_construcion}

\noindent
In this section, we present a new construction method for ST codes based on finite-field rank-metric codes in combination with Eisenstein integers.
The construction is similar to the one in \cite{bossert2002space}, but uses a different embedding of the finite-field elements into the complex numbers.
We give a proof that this mapping is rank-distance-preserving, which implies that the spacial diversity order of the ST code is lower-bounded by the minimum rank distance of the finite-field code.
Furthermore, we present simulation results that show a coding gain compared to the codes constructed in \cite{bossert2002space}.

\subsection{Code Construction}
\label{subsec:new_ST_construcion}

Let $\Fq$ be a finite field which is isomorphic to an Eisenstein integer constellation $\EC{\Eprime} \subseteq \CC$ with modulo arithmetic $\oplus$ and $\otimes$, cf.~Theorem~\ref{thm:isomorphism}.
We choose an isomorphism\footnote{E.g., if $\Eprime$ is of Type~$\mathrm{II}$, $q$ is a prime, we can write $\Fq = \{0,\dots,q-1\}$, and $\varphi(x) = \Emod{\Eprime\EI}{x}$ for all $x \in \Fq$ is an automorphism.} $\varphi : \Fq \to \EC{\Eprime}$ and extend the mapping to matrices by applying it entry-wise
\begin{align}
\Phi : \Fq^{m \times n} &\to \CC^{m \times n}, \\
[x_{ij}]_{i,j} &\mapsto [\varphi(x_{ij})]_{i,j}.
\end{align}
We can also define a generalized inverse
\begin{align}
\Phi^{-1} : \EI &\to \Fq^{m \times n}, \\
[x_{ij}]_{i,j} &\mapsto [\varphi^{-1}(\Emod{\Eprime\EI}{x_{ij}})]_{i,j}.
\end{align}
The following theorem lays the foundation for a new class of ST codes based on Eisenstein integers.
\begin{theorem}\label{thm:Phi_rank_preserving}
The mapping $\Phi$ is minimum rank-distance-preserving, i.e., for any rank-metric code $\Ccode \subseteq \Fq^{m \times n}$ of minimum distance $d$ the code $\CcodeE = \Phi(\Ccode) \subseteq \CC^{m \times n}$ has minimum distance $d$.
\end{theorem}

\begin{proof}
W.l.o.g., $n \leq m$;
otherwise transpose all matrices.
Let $\Cone,\Ctwo \in \Ccode$, $\Cone \neq \Ctwo$.
Then, $\rk(\Cone-\Ctwo) \geq d$.
Take $d$ linearly independent columns of $\Cone-\Ctwo$, w.l.o.g. $\cone_1-\ctwo_1,\dots,\cone_d-\ctwo_d \in \Fq^m$.
We can expand this set of vectors to a basis $\cone_1-\ctwo_1,\dots,\cone_d-\ctwo_d,\ct_{d+1},\dots,\ct_m$ of $\Fq^m$ and define the matrices
\begin{align}
\Ctone &= [\cone_1,\dots,\cone_d,\ct_{d+1},\dots,\ct_m] \in \Fq^{m \times m}, \\
\Cttwo &= [\ctwo_1,\dots,\ctwo_d,\zero,\dots,\zero] \in \Fq^{m \times m}.
\end{align}
Thus, $\rk(\Ctone-\Cttwo) = m$ and $\det(\Ctone-\Cttwo) \neq 0$.

Since $\varphi : \Fq \to \EI$ is an isomorphism, we know that $\Emod{\Eprime\EI}{\Phi(\Ctone)-\Phi(\Cttwo)} = \Phi(\Ctone-\Cttwo)$, implying that there is an $\A \in (\Eprime\EI)^{m \times m}$ such that $\Phi(\Ctone)-\Phi(\Cttwo) = \A + \Phi(\Ctone-\Cttwo)$.
It follows from Lemmas~\ref{lem:mod_det} $\mathrm{(*)}$ and \ref{lem:det_invariant_under_phi} $\mathrm{(**)}$ (in the appendix) that
\begin{align}
a &:= \Emod{\Eprime\EI}{\det(\Phi(\Ctone)-\Phi(\Cttwo))} \\
&= \Emod{\Eprime\EI}{\det(\A + \Phi(\Ctone-\Cttwo))} \\
&\overset{\mathrm{(*)}}{=} \Emod{\Eprime\EI}{\det(\Phi(\Ctone-\Cttwo))} \\
&\overset{\mathrm{(**)}}{=} \varphi(\underset{\neq 0}{\underbrace{\det(\Ctone-\Cttwo)}}) \neq 0.
\end{align}
Thus, $\det(\Phi(\Ctone)-\Phi(\Cttwo)) = a + b \neq 0$, for some $b \in \Eprime\EI$ ($\Eprime \nmid a$ but $\Eprime \mid b$), and $\rk(\Phi(\Ctone)-\Phi(\Cttwo))=m$ (full rank).
Hence, the first $d$ columns of $\Phi(\Cone)-\Phi(\Cone)$ are linearly independent and
\begin{align}
\rk(\Phi(\Cone)-\Phi(\Ctwo)) \geq d,
\end{align}
proving the claim.
\end{proof}

It can be shown (one of the two design criteria in \cite{tarokh1998space}) that the diversity order of an ST code is lower-bounded by its rank distance.
Since the mapping $\Phi$ is rank-distance-preserving, we can design the diversity order of the ST code by choosing the finite-field rank-metric code accordingly.
\begin{example}
\label{ex:full_diversity_gabidulin_ST_code}
We can take a Gabidulin code $\CGab[n,k]$ over the field $\Fqm$ with minimum distance $d = n-k+1$ and obtain an ST code $\CST = \Phi(\CGab)$ with spatial diversity $d$.
The resulting codewords $\X \in \CST$ are complex matrices of dimension $m \times n$ (we must choose $m = \Ntx$ and $n = \Ntime$. If $\Ntx>\Ntime$, we transpose the codewords and set $m = \Ntime$ and $n = \Ntx$).
Since $k$ can be chosen $1\leq k < n$, we are flexible in the tradeoff between code rate $\tfrac{k}{n}$ and diversity $d=n-k+1$.
In the special case of $k=1$ (rank-metric repetition code equivalent), the resulting code $\CST$ has maximum diversity $n$.
\end{example}

\subsection{Numerical Results}

Figure~\ref{fig:simulation_ML} shows simulation results (\emph{frame error rate} (FER) over SNR) of a comparison of Gaussian integer ST codes from finite-field Gabidulin codes \cite{bossert2002space} and our construction presented in Section~\ref{subsec:new_ST_construcion}.
We use the channel model described in Section~\ref{subsec:channel_model} with $\Ntx=\Nrx=\Ntime=4$.
The ST codes are defined using a $\CGab[4,1]$ code of minimum distance $d=4$.
As usual, we ML-decode ST codes by determining 
\begin{align}
\hat{\X} = \argmin\limits_{\X' \in \Ccode} \|\H \X' - \Y\|_\mathrm{F}, \label{eq:ML_dec}
\end{align}
where $\|\cdot\|_\mathrm{F}$ is the Frobenius norm, using exhaustive search.

We use $q=13$ and $q=37$ as the field size since for both, there are Gaussian and Eisenstein primes whose constellations are isomorphic to $\F_{13}$ and $\F_{37}$ respectively, cf.~\cite{stern2015lattice}.

\begin{figure}[b]
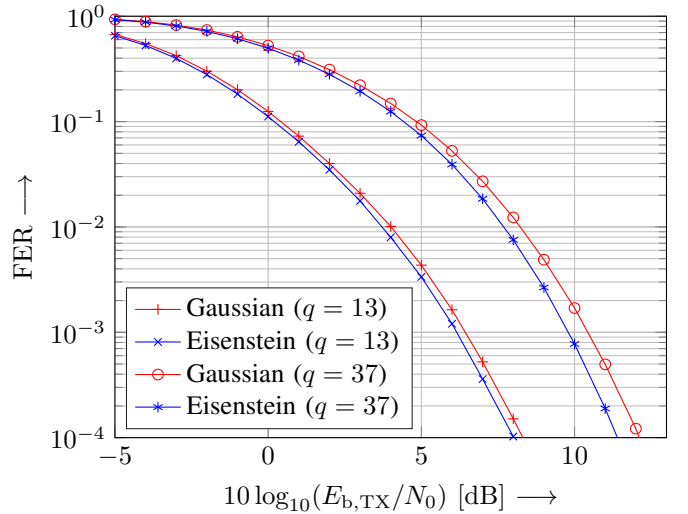

\include{result_ML}
\caption{ML decoding results for $\Ntx=\Nrx=\Ntime=4$, using a Gabidulin code $\CGab[4,1]$ over $\F_{q^4}$, mapped to Gaussian and Eisenstein integer constellations with $q \in \{13,37\}$.
I.i.d.\ unit-variance complex-Gaussian channel matrix $\H$, additive i.i.d.\ complex-Gaussian noise matrix $\N$.}
\label{fig:simulation_ML}
\end{figure}

In both scenarios, our construcion provides a coding gain compared to the Gaussian integer ST codes from~\cite{bossert2002space}.
At $\mathrm{FER}= 10^{-3}$, for $q=13$, the gain is approximately $0.3$ dB and in the $q=37$ case, we are more than $0.6$ dB better.
This gain is expected since Eisenstein integers are more densely packed in the complex plane than Gaussian integers, cf.~\cite{conway2013sphere}.

\section{Alternative Decoding}
\label{sec:alternative_decoding}

\noindent
The complexity of the ML-decoding method used above is proportional to the number of ST codewords.
For instance, the ST code constructed in Example~\ref{ex:full_diversity_gabidulin_ST_code} has $q^\Ntx$ codewords and ML decoding is not possible in sufficiently short time already for small field sizes $q$ or transmit antenna numbers $\Ntx$.

It is interesting to note that although rank-metric codes have been used before to construct new ST codes, to the best of our knowledge, their decoding has not yet been employed.\footnote{In \cite{robert2015new} a decoder of a generalized Gabidulin code is used. In their channel model, $\H$ is always the identity matrix and $\N$ naturally contains criss-cross error patterns. Hence, it differs significantly from the channel model for which ST codes were originally designed \cite{tarokh1998space}.}
In this section, we propose a new decoding scheme which utilizes the decoding capabilities of Gabidulin codes in combination with a channel transformation based on LRA equalization.
For simplicity, we assume $\Ntx=\Nrx$, implying that $\H$ is invertible with probability $1$ (see, e.g., \cite{fischer2016factorization} on how LRA equalization works if $\Ntx \neq \Nrx$).

\subsection{Channel Transformation using LRA Techniques}
\label{subsec:channel_transformation}

In LRA zero-forcing linear equalization \cite{windpassinger2003low,fischer2016factorization}, the inverse channel matrix $\H^{-1}$ is decomposed\footnote{See \cite{fischer2016factorization} for an overview of different factorization criteria.} into
\begin{align}
\H^{-1} = \Z \Fred
\end{align}
such that $\Z \in \EI$, $\det \Z \in \EIunits$ (implying $\Z^{-1} \in \EI$), and the maximum of the row norms
\begin{align}
\max_{i} \|\f_i\|_2 \quad (\f_i \text{ is the $i$th row of $\Fred$})
\end{align}
is minimal among all decompositions.
The problem of finding such a decomposition is equivalent to solving the \emph{shortest basis problem} (SBP) in an $\EI$-lattice\footnote{The same decomposition is possible for Gaussian integers. However, it performs better (in terms of $\max_{i} \|\f_i\|_2$) for Eisenstein integers, cf.\ \cite{stern2015lattice}.} (with the rows of $\H^{-1}$ forming a basis of the lattice).
The SBP is NP-hard.
However, we can find an approximate solution using the LLL algorithm\footnote{For Eisenstein integers, the LLL algorithm has to be adapted, cf.~\cite{stern2015lattice,stern2016advanced}.} in time $\BigO(m^4)$.
Since we know $\Fred$, we can compute the alternative receive matrix
\begin{align}
\tilde{\Y} = \Fred \Y = \Z^{-1} \X + \Fred \N.
\end{align}
Due to $\Z^{-1} \X \in \EI^{m \times n}$, we can make a component-wise decision of the entries of $\tilde{\Y}$ to the closest point in $\EI$ using the quantization function and obtain
\begin{align}
\hat{\Y} = \Equant{\EI}{\tilde{\Y}} = \Z^{-1} \X + \Equant{\EI}{\Fred \N} =: \Z^{-1} \X + \E.
\end{align}
Since $\hat{\Y} \in \EI^{m \times n}$, we can use the generalized inverse of $\Phi$ to get back to finite fields
\begin{align}
\hat{\Y}_\F &= \Phi^{-1}(\hat{\Y}) = \Phi^{-1}(\Z^{-1}) \Phi^{-1}(\X) + \Phi^{-1}(\E) \\
&=: \Z_\F^{-1} \X_\F^{\vphantom{-1}} + \E_\F^{\vphantom{-1}}.
\end{align}
Also, $\det(\Z_\F^{-1}) = \det (\Phi^{-1}(\Z^{-1})) = \varphi^{-1}(\det(\Z)^{-1}) \neq 0$ (since $\det(\Z)$ is a unit in $\EI$) and thus, $\Z_\F^{-1}$ is invertible and we can compute
\begin{align}
\bar{\Y}_\F = \Z_\F\hat{\Y}_\F = \X_\F + \Z_\F\E_\F.
\end{align}
We have transformed the MIMO fading channel, which can be seen as a \emph{multiplicative additive matrix channel} over $\CC$, into an \emph{additive matrix channel} over $\Fq$.
Figure~\ref{fig:illustration_channel_trafo} illustrates the channel transformation procedure.

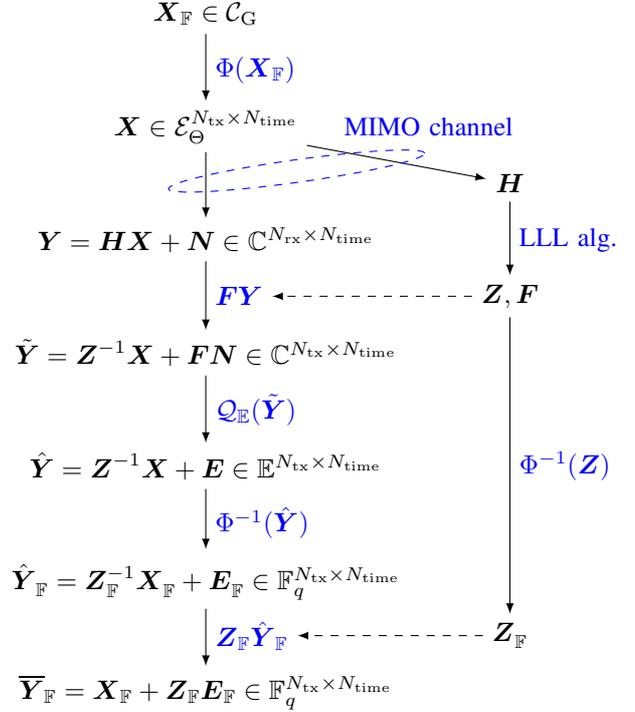
\begin{figure}[ht]
\begin{tikzpicture}
\def\ydif{1.5}
\def\xdif{2}
\node (Z) at (0,\ydif) {$\X_\F \in \CGab$};
\node[right] (A1) at (0,0.5*\ydif) {\labelcolor{$\Phi(\X_\F)$}};
\node (A) at (0,0) {$\X \in \EC{\Eprime}^{\Ntx \times \Ntime}$};
\node (B) at (0,-\ydif) {$\Y = \H \X + \N \in \CC^{\Nrx \times \Ntime}$};
\node[right] (C1) at (0,-1.5*\ydif) {\labelcolor{$\Fred \Y$}};
\node (C) at (0,-2*\ydif) {$\tilde{\Y} = \Z^{-1} \X + \Fred \N \in \CC^{\Ntx \times \Ntime}$};
\node[right] (D1) at (0,-2.5*\ydif) {\labelcolor{$\Equant{\EI}{\tilde{\Y}}$}};
\node (D) at (0,-3*\ydif) {$\hat{\Y} = \Z^{-1} \X + \E \in \EI^{\Ntx \times \Ntime}$};
\node[right] (E1) at (0,-3.5*\ydif) {\labelcolor{$\Phi^{-1}(\hat{\Y})$}};
\node (E) at (0,-4*\ydif) {$\hat{\Y}_\F^{\vphantom{-1}} = \Z_\F^{-1} \X_\F^{\vphantom{-1}} + \E_\F^{\vphantom{-1}} \in \Fq^{\Ntx \times \Ntime}$};
\node[right] (F1) at (0,-4.5*\ydif) {\labelcolor{$\Z_\F \hat{\Y}_\F$}};
\node (F) at (0,-5*\ydif) {$\bar{\Y}_\F = \X_\F + \Z_\F \E_\F \in \Fq^{\Ntx \times \Ntime}$};

\node (A2) at (2*\xdif,-0.5*\ydif) {$\H$};
\node[right] (B3) at (2*\xdif,-1*\ydif) {\labelcolor{LLL alg.}};
\node (B2) at (2*\xdif,-1.5*\ydif) {$\Z,\Fred$};
\node[right] (F3) at (2*\xdif,-3*\ydif) {\labelcolor{$\Phi^{-1}(\Z)$}};
\node (F2) at (2*\xdif,-4.5*\ydif) {$\Z_\F^{}$};

\draw[->,>=latex] (Z) -- (A);
\draw[->,>=latex] (A) -- (B);
\draw[->,>=latex] (B) -- (C);
\draw[->,>=latex] (C) -- (D);
\draw[->,>=latex] (D) -- (E);
\draw[->,>=latex] (E) -- (F);

\draw[->,>=latex] (A) -- (A2);
\draw[->,>=latex] (A2) -- (B2);
\draw[->,>=latex] (B2) -- (F2);

\draw[->,>=latex,dashed] (B2) -- (C1);
\draw[->,>=latex,dashed] (F2) -- (F1);

\draw[rotate=8,dashed,darkgreen] (1.1,-0.75) ellipse (1.7 and 0.15);
\node[right] (A1) at (1.7,0) {\labelcolor{MIMO channel}};
\end{tikzpicture}
\caption{Illustration of the channel transformation.}
\label{fig:illustration_channel_trafo}
\end{figure}

\subsection{Decoding Using Rank-Metric Decoder}
\label{subsec:FN_properties_cris_cross_patterns}

In order to see how rank-metric codes can be used to correct errors of the form $\Z_\F^{}\E_\F^{}$, we have a closer look at the error matrix $\E$.
An entry of $\E$ is non-zero if the corresponding entry in $\Fred \N$ is large enough (by absolute value) to be closer to some element of $\EI \setminus \{0\}$ than to $0$.
It can be observed that the rows of $\Fred$ have different norms $\|\f_i\|_2$. Since the entries of $\N$ are i.i.d. $\Norm(0,\sigmaN)$ distributed for some noise variance $\sigmaN$, an entry in the $i$th row of $\Fred \N$ is $\Norm(0,\|\f_i\|_2^2 \sigmaN)$ distributed (and i.i.d. to other entries in that row).
Thus, those rows of $\E$ with larger $\|\f_i\|_2$ tend to contain more errors than others. Since the $\|\f_i\|_2$'s might differ a lot,\footnote{Finding the distribution of the row norms of $\Fred$ is an open problem and is beyond the scope of this paper since it would involve a detailed analysis of the numerical properties of the LLL algorithm.} in general, non-zero entries of $\E$ tend to occur row-wise.

Also, entries in columns are no longer independent and thus, if there is a relatively large entry in $\N$, this value might influence the entries of the entire column in $\Fred \N$, or $\E$.

We can thus conclude that $\E$ tends to contain criss-cross error patterns and therefore has low rank.
We cannot use arbitrary criss-cross error correcting codes because the multiplication by $\Z_\F^{}$ in the final error matrix destroys the criss-cross pattern.
However, the rank is preserved, meaning that the matrix $\Z_\F \E_\F$ tends to have low rank and can be corrected using a rank-metric code.

\begin{example}\label{ex:ex1}
Let $\Ntx=\Nrx=\Ntime=7$ and $6~\mathrm{dB}$ SNR.
A realistic output of the channel matrix decomposition is $\Fred \in \CC^{7 \times 7}$ with squared row norms:
\begin{center}
\begin{tabular}{c|c|c|c|c|c|c|c}
\hline
$i$ & $1$ & $2$ & $3$ & $4$ & $5$ & $6$ & $7$ \\
\hline
$\|\f_i\|_2^2$ & $0.38$ & $0.21$ & $0.28$ & $0.19$ & $0.20$ & $0.34$ & $0.24$\\
\hline
\end{tabular}
\end{center}
For instance, the error matrix $\E_\F^{}$ in the channel transformation procedure can have the form (here, $*$ means that this entry is non-zero, all other entries are zero)
\begin{align}
\arraycolsep=1pt\def\arraystretch{0.5}
\E_\F =
\begin{bmatrix}
*&*& &*&*&*&*\\
 & & & & &*& \\
*&*&*& &*&*& \\
 & & & & & & \\
 & & & & &*& \\
*& &*& &*&*&*\\
 & & & & &*& \\
\end{bmatrix}
\; \Rightarrow
\; \rk(\E_\F) \leq 4.
\end{align}
The rows which contain many errors ($i=1$, $3$ and $6$) are due to large values of $\|\f_i\|_2^2$ and the corrupted column ($j=6$) results from a large value in the $j$th column of the original noise matrix $\N$, which spreads through the entire column due to the matrix multiplication $\Fred \N$.
\end{example}

\subsection{Improved Decoding Using GMD}
\label{subsec:gmd}

Since we know $\Fred$, the squared row norms $\|\f_i\|_2^2$ provide reliability information of the rows of $\E$.
Thus, we can use \emph{generalized minimum distance} (GMD) decoding \cite{bossert1999channel} in combination with an error-and-erasure decoding algorithm for Gabidulin codes (cf. Section~\ref{subsec:Gabidulin}) to obtain better results.

More exactly, we can start by trying to decode without erasures.
Then, incrementally from $\ell=1$ to $d-1$, we estimate the likeliest $\ell$ rows of $\E_\F$ which are in error, using the soft information given by the $\|\f_i\|_2$'s, say $\Eps_\ell \subseteq \{1,\dots,m\}$, $|\Eps_\ell| = \ell$ (e.g., $\Eps_2 = \{1,6\}$ in Example~\ref{ex:ex1}).
Then we can decompose the error into
\begin{align}
\Z_\F^{} \E_\F^{} = \Z_\F^{} \E_\F' + \Z_\F^{} \E_\F'',
\end{align}
where $\E_\F''$ contains non-zero values only in the rows $\Eps_\ell$ and $\E_\F'$ has zero rows in $\Eps_\ell$.
We can re-write $\Z_\F^{}\E_\F'' = [\Z_\F^{}]_{\Eps_\ell} \EFell$, where $\ZFell \in \Fq^{m \times \ell}$ consists of the columns of $\Z_\F^{}$ with indices in $\Eps_\ell$ and the rows of $\EFell \in \Fq^{\ell \times m}$ are the non-zero rows of $\E_\F''$.
The procedure is illustrated in the following example.

\begin{example}
Let $\E_\F^{}$ be as in Example~\ref{ex:ex1} and $k=1$. Thus, our finite-field Gabidulin code has parameters $[7,1]$, minimum rank distance $7$, and we cannot correct the rank error with a half-the minimum rank distance decoder since $\rk(\E_\F) = 4 > 3 = \frac{d-1}{2}$.
Using GMD, we can, e.g., declare $\ell=2$ erasures as follows (recall that $\Eps_2 = \{1,6\}$):
\begin{align}
\Z_\F^{}
\arraycolsep=1pt\def\arraystretch{0.5}
\begin{bmatrix}
*&*& &*&*&*&*\\
 & & & & &*& \\
*&*&*& &*&*& \\
 & & & & & & \\
 & & & & &*& \\
*&\phantom{*}&*&\phantom{*}&*&*&*\\
 & & & & &*& \\
\end{bmatrix}
&=
\Z_\F^{}
\left(
\arraycolsep=1pt\def\arraystretch{0.5}
\begin{bmatrix}
 & & & & &\phantom{*}&\phantom{*}\\
 & & & & &*& \\
*&*&*&\phantom{*}&*&*& \\
 & & & & &\phantom{*}& \\
 & & & & &*& \\
 & & & & &\phantom{*} \\
 & & & & &*& \\
\end{bmatrix}
+
\arraycolsep=1pt\def\arraystretch{0.5}
\begin{bmatrix}
*&*& &*&*&*&*\\
 & & & & & & \\
 & & & & & & \\
 & & & & & & \\
 & & & & & & \\
*& &*& &*&*&*\\
 & & & & & & \\
\end{bmatrix}
\right) \notag \\
&=
\Z_\F^{} \E_\F'
+
[\Z_\F^{}]_{\Eps_2}
\arraycolsep=1pt\def\arraystretch{0.5}
\begin{bmatrix}
*&*& &*&*&*&*\\
*& &*& &*&*&*\\
\end{bmatrix},
\end{align}
where $\rk(\Z_\F^{} \E_\F') = 2$ (unknown) and $[\Z_\F^{}]_{\Eps_2} \in \Fq^{7 \times 2}$ (known) consists of the columns of $\Z_\F^{}$ with indices $\Eps_2$.
Thus, we can correctly decode due to \eqref{eq:dec_cond} and
\begin{align}
2 \cdot \rk(\Z_\F^{} \E_\F') + \rk(\ZFell) = 4 + 2 < 7 = d.
\end{align}
\end{example}

If we use a Gabidulin code of dimension $1$ as in \cite{bossert2002space} or Example~\ref{ex:full_diversity_gabidulin_ST_code}, we need to know only one row in $\Z_\F^{} \X_\F^{}$ which does not contain an error for decoding successfully.
Since there are only as many possibilities as there are rows, we can simply ``try'' all rows, meaning that iteratively for each row $i$ we declare an erasure in all other rows than the $i$th one, decode and obtain a candidate codeword.
Among these candidates, we then find the one with minimum Frobenius norm difference to the received word as in \eqref{eq:ML_dec}.
We call this method \emph{multi-trial} (MT) GMD decoding here.

\subsection{Numerical Results}
\label{subsec:alt_dec_numerical_results}

Figure~\ref{fig:simulation_MMSE_pure} shows simulation results.
We use ST codes based on a $\CGab[4,1]$ code of minimum distance $d=4$ with an Eisenstein integer constellation of size $q=13$, and the channel model described in Section~\ref{subsec:channel_model} with $\Ntx=\Nrx=\Ntime=4$.
We compare ML decoding to the alternative decoding methods described in this section; BMD as in Section~\ref{subsec:FN_properties_cris_cross_patterns} and both GMD and MT GMD as in Section~\ref{subsec:gmd}.

For comparison, we perform factorization and equalization based on both zero-forcing (ZF) linear equalization (as described in Section~\ref{subsec:channel_transformation}) and the \emph{minimum mean-squared error} (MMSE) criterion.
The latter is not described here in detail for reasons of clarity, but can, e.g., be found in \cite{fischer2016factorization}.

\begin{figure}[ht!]
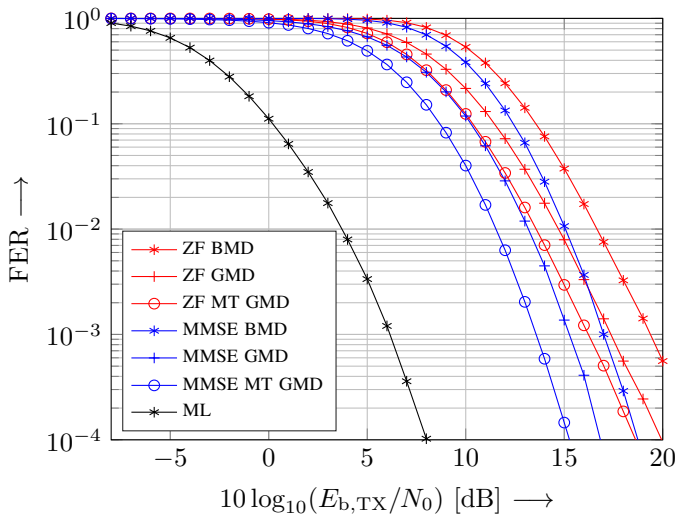

\include{result_alt_dec_MMSE}
\caption{Comparison of ML decoding and alternative decoders (BMD, GMD, MT GMD) based on LRA equalization in the case $q=13$, $\Ntx=\Nrx=\Ntime=4$, and $\CGab[4,1]$. ZF and MMSE indicates that the ZF or the MMSE criterion, respectively, was used for both factorization and equalization.
I.i.d.\ unit-variance complex-Gaussian channel matrix $\H$, additive i.i.d.\ complex-Gaussian noise matrix $\N$.}
\label{fig:simulation_MMSE_pure}
\end{figure}

It can be seen that all alternative decoding methods are suboptimal compared to the ML case.
The best of the alternatives, multi-trial GMD with MMSE factorization and equalization, is approximately $7$ dB worse than ML decoding at FER $10^{-3}$.
This effect can be expected due to the following reasons.
\begin{itemize}
\item The row norms of $\Fred$ do not provide actual soft information. They merely describe a statistical tendency of the errors in $\Fred \Y$.
\item GMD decoding of Gabidulin codes cannot fully utilize soft information. To our knowledge, there is no soft-information decoding algorithm for Gabidulin codes, yet.
\item The LLL algorithm only finds an approximate solution to the shortest basis problem.
\end{itemize}

However, all alternative decoding methods share the advantage that their decoding has polynomial decoding complexity in the parameters $\Ntx$, $\Nrx$, and $\Ntime$ of the code.
It can therefore be used for larger parameter sets.

\section{Conclusion}\label{sec:Conclusion}

\noindent
We have presented a new class of space-time codes based on finite-field rank-metric codes and Eisenstein integers.
These codes achieve maximum diversity order and improve upon existing ST codes based on Gaussian integers.
We have also shown how to decode the new code class in polynomial time using a channel transformation based on lattice-reduction-aided equalization.

In future work, the problems causing the sub-optimality of the alternative decoder, as discussed in Section~\ref{subsec:alt_dec_numerical_results}, should be solved in order to reduce the gap to ML decoding.
Alternatively, a modification of the code construction using concatenation with Hamming-error-correcting codes in the rows can be considered, which could shift all curves (including ML decoding) to lower SNR values since the probability of a row being in error decreases.
However, if long Hamming-error correcting codes of large dimension are used, ML decoding becomes impractical due to the large number of codewords, resulting in an advantage for our alternative decoding method.

\appendix

\subsection{Technical Proofs}
\label{app:technical_proofs}

We choose $\Eprime$ and $\varphi$ as in Section~\ref{sec:new_ST_construcion}.

\begin{lemma}\label{lem:mod_det}
Let $\A \in (\Eprime \EI)^{m \times n}$, $\B \in \EI^{m \times n}$.
Then,
\begin{align}
\Emod{\Eprime \EI}{\det (\A + \B)} = \Emod{\Eprime \EI}{\det(\B)}.
\end{align}
\end{lemma}

\begin{proof}
For $a,b \in \Eprime$ it holds that
\begin{align}
\Emod{\Eprime \EI}{a+b} &= \Emod{\Eprime \EI}{\Emod{\Eprime \EI}{a}+\Emod{\Eprime \EI}{b}}, \label{eq:mod_plus_invariant}\\
\Emod{\Eprime \EI}{a \cdot b} &= \Emod{\Eprime \EI}{\Emod{\Eprime \EI}{a} \cdot \Emod{\Eprime \EI}{b}}. \label{eq:mod_mul_invariant}
\end{align}
The determinant is a finite sum of finitely many multiplications of matrix elements, so this relation extends to $\det$ as follows:
\begin{align}
&\Emod{\Eprime \EI}{\det(\A+\B)} \\
&= \Emod{\Eprime \EI}{\det(\Emod{\Eprime \EI}{\B})} = \Emod{\Eprime \EI}{\det(\B)},
\end{align}
which proves the claim (note that $\Emod{\Eprime \EI}{\A}={\ve 0}$).
\end{proof}

\begin{lemma}\label{lem:det_invariant_under_phi}
For any $\A \in \Fq^{m \times n}$, 
\begin{align}
\varphi(\det (\A)) = \Emod{\Eprime \EI}{\det(\Phi(\A))}.
\end{align}
\end{lemma}

\begin{proof}
Since $\varphi : \Fq \to (\EI,\oplus,\otimes)$ is an isomorphism, $\varphi(\det (\A)) = \det_{\oplus,\otimes}(\Phi(\A))$, where $\det_{\oplus,\otimes}$ is the determinant under modulo addition $\oplus$ and multiplication $\otimes$.
We obtain
%\begin{align}
$\Emod{\Eprime \EI}{\det(\Phi(\A))} = \det\nolimits_{\oplus,\otimes}(\Phi(\A)) = \varphi(\det(\A))$,
%\end{align}
where the first equality follows by \eqref{eq:mod_plus_invariant} and \eqref{eq:mod_mul_invariant}.
\end{proof}

\bibliographystyle{IEEEtran}
\bibliography{main}

\end{document}

%% file: preamble.tex
\usepackage[utf8]{inputenc}
\usepackage[T1]{fontenc}
\usepackage{url}
\usepackage[cmex10]{amsmath}
\usepackage{cite}
\usepackage{amsmath}
\usepackage{amssymb}
\usepackage{amsthm}
\usepackage{mathrsfs}
\usepackage{graphicx}
\usepackage{color}
\usepackage[linesnumbered,ruled,vlined,titlenumbered]{algorithm2e}
\usepackage{tikz}
\usepackage{pgfplots}
\pgfplotsset{compat=newest}

\IEEEoverridecommandlockouts

\makeatletter
\def\thmhead@plain#1#2#3{%
  \thmname{#1}\thmnumber{\@ifnotempty{#1}{ }\@upn{#2}}%
  \thmnote{ {\the\thm@notefont#3}}}
\let\thmhead\thmhead@plain
\makeatother

\newtheorem{theorem}{Theorem}
\newtheorem{lemma}{Lemma}

\newtheorem{example}{Example}

\usepackage{varioref}        % for \vref{:...} gives "As 1.2 on page 4"
\usepackage{fancyref}

\makeatletter
\def\mkfancyprefix#1#2{%
\expandafter\def\csname fancyref#1labelprefix\endcsname{#1}%
\begingroup\def\x{\endgroup\frefformat{plain}}%
    \expandafter\x\csname fancyref#1labelprefix\endcsname
    {\MakeLowercase{#2}\fancyrefdefaultspacing##1}%
\begingroup\def\x{\endgroup\Frefformat{plain}}%
    \expandafter\x\csname fancyref#1labelprefix\endcsname
    {#2\fancyrefdefaultspacing##1}%
\begingroup\def\x{\endgroup\frefformat{vario}}%
    \expandafter\x\csname fancyref#1labelprefix\endcsname
    {\MakeLowercase{#2}\fancyrefdefaultspacing##1##3}%
\begingroup\def\x{\endgroup\Frefformat{vario}}%
    \expandafter\x\csname fancyref#1labelprefix\endcsname
    {#2\fancyrefdefaultspacing##1##3}%
}
\makeatother
\fancyrefchangeprefix{\fancyrefeqlabelprefix}{eqn}
\mkfancyprefix{ssec}{Section}
\mkfancyprefix{tbl}{Table}
\mkfancyprefix{thm}{Theorem}
\mkfancyprefix{lem}{Lemma}
\mkfancyprefix{cor}{Corollary}
\mkfancyprefix{prop}{Proposition}
\mkfancyprefix{prob}{Problem}
\mkfancyprefix{alg}{Algorithm}
\mkfancyprefix{inv}{Invariant}
\mkfancyprefix{ex}{Example}
\mkfancyprefix{line}{Line}
\mkfancyprefix{def}{Definition}
\mkfancyprefix{itm}{Item}
\mkfancyprefix{app}{Appendix}
\mkfancyprefix{rem}{Remark}
\newcommand{\cref}[1]{\Fref{#1}}

\makeatletter
\newcommand{\removelatexerror}{\let\@latex@error\@gobble}
\makeatother

\DeclareMathOperator{\rk}{rank}

\DeclareMathOperator{\dR}{\ensuremath{d_R}}

\def\argmin{\mathop{\mathrm{argmin}}}

\def\ve#1{{\mathchoice{\mbox{\boldmath$\displaystyle #1$}}%
              {\mbox{\boldmath$\textstyle #1$}}%
              {\mbox{\boldmath$\scriptstyle #1$}}%
              {\mbox{\boldmath$\scriptscriptstyle #1$}}}}

\newcommand{\Fqm}{\ensuremath{\mathbb F_{q^m}}}

\newcommand{\Fq}{\ensuremath{\mathbb F_{q}}}
\newcommand{\F}{\ensuremath{\mathbb F}}

\newcommand{\ZZ}{\ensuremath{\mathbb{Z}}}
\newcommand{\NN}{\ensuremath{\mathbb{N}}}

\renewcommand{\c}{\ve c}

\renewcommand{\r}{\ve r}

\renewcommand{\i}{\ve i}

\newcommand{\A}{\ve{A}}
\newcommand{\B}{\ve{B}}
\newcommand{\E}{\ve{E}}

\renewcommand{\H}{\ve{H}}

\newcommand{\X}{\ve{X}}
\newcommand{\Y}{\ve{Y}}

\renewcommand{\bar}{\overline}

\newcommand{\w}{\omega}

\newcommand{\imag}{\mathrm{j}}
\newcommand{\eexp}[1]{\mathrm{e}^{#1}}
\newcommand{\CC}{\mathbb{C}}

\newcommand{\KK}{\mathbb{K}}
\newcommand{\EI}{\mathbb{E}}
\newcommand{\GI}{\mathbb{G}}
\newcommand{\EIunits}{\EI^{\times}}

\newcommand{\N}{\ve N}
\newcommand{\Z}{\ve Z}

\newcommand{\Ccode}{\mathcal{C}}
\newcommand{\CcodeE}{\mathcal{C}_\mathbb{E}}
\newcommand{\CGab}{\mathcal{C}_\mathrm{G}}
\newcommand{\CST}{\mathcal{C}_\mathrm{ST}}
\newcommand{\C}{\ve C}
\newcommand{\Norm}{\mathcal{N}}
\newcommand{\sigmaN}{\sigma_n^2}

\newcommand{\AR}{\ve A_\mathrm{R}}
\newcommand{\BR}{\ve B_\mathrm{R}}
\newcommand{\AC}{\ve A_\mathrm{C}}
\newcommand{\BC}{\ve B_\mathrm{C}}

\newcommand{\Eps}{\mathcal{E}}
\newcommand{\ZFell}{[\Z_\F^{}]_{\Eps_\ell}}
\newcommand{\EFell}{[\E_\F'']_{\Eps_\ell}}
\newcommand{\Fred}{\ve F}
\newcommand{\f}{\ve f}

\newcommand{\Voronoi}[1]{\mathcal{R}_\mathrm{V}(#1)}
\newcommand{\Eprime}{\Theta}
\newcommand{\Equant}[2]{\mathcal{Q}_{#1}(#2)}
\newcommand{\Emod}[2]{\mathrm{mod}_{#1}(#2)}
\newcommand{\EC}[1]{\mathcal{E}_{#1}}
\newcommand{\Eadd}{\oplus}
\newcommand{\Emul}{\otimes}
\newcommand{\ext}{\mathrm{ext}}
\newcommand{\Lset}{\mathcal{L}_{q^m}}

\newcommand{\BigO}{\mathcal{O}}
\newcommand{\Otilde}{\BigO^{\sim}}

\newcommand{\Cone}{\C^{(1)}}
\newcommand{\Ctwo}{\C^{(2)}}
\newcommand{\Ctone}{\tilde{\C}^{(1)}}
\newcommand{\Cttwo}{\tilde{\C}^{(2)}}
\newcommand{\cone}{\c^{(1)}}
\newcommand{\ctwo}{\c^{(2)}}

\newcommand{\ct}{\tilde{\c}}
\newcommand{\zero}{\ve 0}

\newcommand{\Ntx}{{N_{\mathrm{tx}}}}
\newcommand{\Nrx}{{N_{\mathrm{rx}}}}
\newcommand{\Ntime}{{N_{\mathrm{time}}}}

\definecolor{darkgreen}{rgb}{0,0,1} %{0,0.5,0}
\newcommand{\labelcolor}[1]{\textcolor{darkgreen}{#1}}

%% file: result_ML.tex
\begin{tikzpicture}

\begin{axis}[%
width=0.4\textwidth,
height=2.2in,
at={(0.758in,0.481in)},
scale only axis,
xmin=-5,
xmax=13,
xlabel={$10 \log_{10}(E_\mathrm{b,TX}/N_0)$ $\mathrm{[dB]}$ $\longrightarrow$},
xmajorgrids,
ymode=log,
ymin=1e-04,
ymax=1,
yminorticks=true,
ylabel={$\mathrm{FER}$ $\longrightarrow$},
ymajorgrids,
yminorgrids,
axis background/.style={fill=white},
title style={font=\bfseries},
legend style={at={(0.02,0.02)},anchor=south west, legend cell align=left,align=left,draw=white!15!black}
]
\addplot [color=red,solid,mark=+]
  table[row sep=crcr]{%
-8    0.908521127709156\\
-7    0.853307560716066\\
-6    0.774463556890583\\
-5    0.672100917307615\\
-4    0.551023047185442\\
-3    0.422381766995784\\
-2    0.301369899917833\\
-1    0.200781260523445\\
0    0.124869677662684\\
1    0.0727824997642748\\
2    0.0400274788184108\\
3    0.0208865959940193\\
4    0.010111935774997\\
5    0.00433734290601975\\
6    0.00163660609652608\\
7    0.000523983351068845\\
8    0.000150864101078948\\
9    4.04100270747181e-05\\
10    4.04100270747181e-06\\
11    1.3470009024906e-06\\
12    0\\
13    0\\
14    0\\
15    0\\
16    0\\
17    0\\
18    0\\
19    0\\
20    0\\
21    0\\
22    0\\
23    0\\
24    0\\
25    0\\
26    0\\
27    0\\
};
\addlegendentry{Gaussian ($q=13$)};

\addplot [color=blue,solid,mark=x]
  table[row sep=crcr]{%
-8	0.903494120341061\\
-7	0.844508950820997\\
-6	0.761462304179744\\
-5	0.653563490887539\\
-4	0.528101132827759\\
-3	0.398305472864667\\
-2	0.279623917348025\\
-1	0.182431067228815\\
0	0.111659639811959\\
1	0.0643772141327335\\
2	0.0348603833564569\\
3	0.0176551408289444\\
4	0.0079715513409394\\
5	0.0033553792481041\\
6	0.0012042188068266\\
7	0.000359649240964991\\
8	0.000102372068589286\\
9	2.69400180498121e-05\\
10	1.3470009024906e-06\\
11	0\\
12	0\\
13	0\\
14	0\\
15	0\\
16	0\\
17	0\\
18	0\\
19	0\\
20	0\\
21	0\\
22	0\\
23	0\\
24	0\\
25	0\\
26	0\\
27	0\\
};
\addlegendentry{Eisenstein ($q=13$)};

\addplot [color=red,solid,mark=o]
  table[row sep=crcr]{%
-6    0.963157512153014\\
-5    0.934157466721185\\
-4    0.889223570033165\\
-3    0.824658579801009\\
-2    0.739559311253464\\
-1    0.637769297169597\\
0    0.526732997137795\\
1    0.415377765662623\\
2    0.311622370632865\\
3    0.221691881332061\\
4    0.148552996229158\\
5    0.0926836581709145\\
6    0.0527118259052292\\
7    0.0271091726863841\\
8    0.012317477624824\\
9    0.00488664758529826\\
10    0.00169915042478761\\
11    0.000496115578574349\\
12    0.000121757303166599\\
13    1.99900049975012e-05\\
14    3.63454636318205e-06\\
15    0\\
};
\addlegendentry{Gaussian ($q=37$)};

\addplot [color=blue,solid,mark=asterisk]
  table[row sep=crcr]{%
-6    0.959330334832584\\
-5    0.927665258279951\\
-4    0.878694289219027\\
-3    0.808789241742765\\
-2    0.718396256417246\\
-1    0.612039434828041\\
0    0.497678433510517\\
1    0.384103402844033\\
2    0.281303893507792\\
3    0.194164735813911\\
4    0.124624051610558\\
5    0.0737058743355595\\
6    0.0392485575394121\\
7    0.0185034755349598\\
8    0.00754350097678433\\
9    0.00265867066466767\\
10    0.000775067011948571\\
11    0.000188996410885466\\
12    3.81627368134115e-05\\
13    8.1777293171596e-06\\
14    9.08636590795511e-07\\
15    0\\
};
\addlegendentry{Eisenstein ($q=37$)};

\end{axis}
\end{tikzpicture}%

%% file: result_alt_dec_MMSE.tex
\begin{tikzpicture}

\begin{axis}[%
width=0.4\textwidth,
height=2.2in,
at={(0.758in,0.481in)},
scale only axis,
xmin=-8,
xmax=20,
xlabel={$10 \log_{10}(E_\mathrm{b,TX}/N_0)$ $\mathrm{[dB]}$ $\longrightarrow$},
xmajorgrids,
ymode=log,
ymin=1e-04,
ymax=1,
yminorticks=true,
ylabel={$\mathrm{FER}$ $\longrightarrow$},
ymajorgrids,
yminorgrids,
axis background/.style={fill=white},
title style={font=\bfseries},
legend style={at={(0.02,0.02)},anchor=south west, legend cell align=left,align=left,draw=white!15!black}
]

\addplot [color=red,solid,mark=asterisk,mark options={solid}]
  table[row sep=crcr]{%
-8	1\\
-7	1\\
-6	1\\
-5	1\\
-4	1\\
-3	1\\
-2	1\\
-1	1\\
0	0.999993204704913\\
1	0.999959507488183\\
2	0.999787484127634\\
3	0.998882872105006\\
4	0.995770440782866\\
5	0.986333544545998\\
6	0.962115671374689\\
7	0.910870954832977\\
8	0.821469911690021\\
9	0.692042160244787\\
10	0.535807155855306\\
11	0.378008488906029\\
12	0.242448123274472\\
13	0.141674455657302\\
14	0.0758438709277449\\
15	0.0375125422064621\\
16	0.0172893716837447\\
17	0.0075582298610027\\
18	0.00326304484885449\\
19	0.00141733099984762\\
20	0.000559262094252232\\
21	0.00021800796017622\\
22	8.0891937400997e-05\\
23	3.36041304968468e-05\\
24	1.08910893854047e-05\\
};
\addlegendentry{\scriptsize ZF BMD};

\addplot [color=red,solid,mark=+,mark options={solid}]
  table[row sep=crcr]{%
-8	0.999945451466839\\
-7	0.999922924598196\\
-6	0.999854692388627\\
-5	0.999734611146686\\
-4	0.999466150447647\\
-3	0.998862486219747\\
-2	0.997505102754636\\
-1	0.994723220649654\\
0	0.989403435459172\\
1	0.979298086975495\\
2	0.961564693955269\\
3	0.931341083134664\\
4	0.8824616878005\\
5	0.810670512245168\\
6	0.712622595733839\\
7	0.591159674819997\\
8	0.458777087046892\\
9	0.328781137228936\\
10	0.216448021065787\\
11	0.130650439129655\\
12	0.0721385733808557\\
13	0.0369930279341549\\
14	0.0176026068613677\\
15	0.00792256938112641\\
16	0.00331582474356837\\
17	0.00140755694527097\\
18	0.000558517404379726\\
19	0.00024407210571394\\
20	9.49479587445534e-05\\
21	3.84446146681378e-05\\
22	1.28459003007337e-05\\
23	4.46813923503781e-06\\
24	1.30320727688603e-06\\
};
\addlegendentry{\scriptsize ZF GMD};

\addplot [color=red,solid,mark=o,mark options={solid}]
  table[row sep=crcr]{%
-8	0.999836261314283\\
-7	0.999794093250252\\
-6	0.999716924762214\\
-5	0.999558585078072\\
-4	0.999163620186941\\
-3	0.998284420706214\\
-2	0.996286231605811\\
-1	0.992039451435376\\
0	0.983773300764322\\
1	0.967844384391412\\
2	0.939943182024452\\
3	0.893005007015444\\
4	0.820188858936842\\
5	0.719758121005961\\
6	0.594127729393105\\
7	0.455495704209927\\
8	0.322521925764287\\
9	0.208725773260365\\
10	0.123973642818997\\
11	0.0674124921912285\\
12	0.0341027934527239\\
13	0.0160320559202519\\
14	0.00704523162507993\\
15	0.00295399855176437\\
16	0.00121998818363345\\
17	0.000505551337197715\\
18	0.000185986295658449\\
19	6.87907269727696e-05\\
20	2.1037488898303e-05\\
21	1.07049169172781e-05\\
22	4.46813923503781e-06\\
23	3.16493195815178e-06\\
24	1.30320727688603e-06\\
};
\addlegendentry{\scriptsize ZF MT GMD};

\addplot [color=blue,solid,mark=asterisk,mark options={solid}]
  table[row sep=crcr]{%
-8	0.999998696792723\\
-7	1\\
-6	0.999998696792723\\
-5	1\\
-4	0.999999348396362\\
-3	0.999994787170892\\
-2	0.999985664719954\\
-1	0.999947127019052\\
0	0.999791673008166\\
1	0.999392239977801\\
2	0.997962993939993\\
3	0.993796361017086\\
4	0.983093398910724\\
5	0.958634525479424\\
6	0.909065919668256\\
7	0.824086938074476\\
8	0.699032191733067\\
9	0.544746878516042\\
10	0.384004545586982\\
11	0.241215475363006\\
12	0.133708042659932\\
13	0.0659352136566442\\
14	0.0281691045485937\\
15	0.0106762463571866\\
16	0.00367253119249889\\
17	0.00100235256839348\\
18	0.000290708308979647\\
19	6.86976407387063e-05\\
20	8.2846748316326e-06\\
21	2.04789714939233e-06\\
22	7.44689872506301e-07\\
23	0\\
24	0\\
};
\addlegendentry{\scriptsize MMSE BMD};

\addplot [color=blue,solid,mark=+,mark options={solid}]
  table[row sep=crcr]{%
-8	0.999775010572269\\
-7	0.999653905381753\\
-6	0.999402479463548\\
-5	0.998620275838714\\
-4	0.996559346616553\\
-3	0.992345518972976\\
-2	0.984000245002968\\
-1	0.971953396935434\\
0	0.9539215694488\\
1	0.928043968724142\\
2	0.889420163055433\\
3	0.836357844707401\\
4	0.763089856048655\\
5	0.668938063864045\\
6	0.55525752446283\\
7	0.430581918352761\\
8	0.308125208804059\\
9	0.20061377339292\\
10	0.118367803631238\\
11	0.0617549901435641\\
12	0.0286981136167753\\
13	0.0119091735273548\\
14	0.00447688934103976\\
15	0.00136846072696439\\
16	0.000409858688580656\\
17	7.56791082934529e-05\\
18	2.04789714939233e-05\\
19	3.63036312846822e-06\\
20	0\\
21	0\\
22	0\\
23	0\\
24	0\\
};
\addlegendentry{\scriptsize MMSE GMD};

\addplot [color=blue,solid,mark=o,mark options={solid}]
  table[row sep=crcr]{%
-8	0.996907861476886\\
-7	0.995343454227218\\
-6	0.993120741130255\\
-5	0.989532359893349\\
-4	0.983995497605031\\
-3	0.975414994721545\\
-2	0.961550451761457\\
-1	0.941481897473519\\
0	0.910768094544337\\
1	0.866162474388486\\
2	0.802453231846997\\
3	0.718026437793681\\
4	0.613345847846641\\
5	0.492086506526695\\
6	0.365273547396299\\
7	0.247505498371131\\
8	0.151314372970662\\
9	0.0822842282038816\\
10	0.0400717620395641\\
11	0.0169707375045461\\
12	0.00629840076919017\\
13	0.00203868161222006\\
14	0.000587932654343725\\
15	0.000145679956309045\\
16	3.55589414121759e-05\\
17	1.11703480875945e-06\\
18	9.30862340632877e-08\\
19	0\\
20	0\\
21	0\\
22	0\\
23	0\\
24	0\\
};
\addlegendentry{\scriptsize MMSE MT GMD};

\addplot [color=black,solid,mark=asterisk,mark options={solid}]
  table[row sep=crcr]{%
-8	0.903494120341061\\
-7	0.844508950820997\\
-6	0.761462304179744\\
-5	0.653563490887539\\
-4	0.528101132827759\\
-3	0.398305472864667\\
-2	0.279623917348025\\
-1	0.182431067228815\\
0	0.111659639811959\\
1	0.0643772141327335\\
2	0.0348603833564569\\
3	0.0176551408289444\\
4	0.0079715513409394\\
5	0.0033553792481041\\
6	0.0012042188068266\\
7	0.000359649240964991\\
8	0.000102372068589286\\
9	2.69400180498121e-05\\
10	1.3470009024906e-06\\
11	0\\
12	0\\
13	0\\
14	0\\
15	0\\
16	0\\
17	0\\
18	0\\
19	0\\
20	0\\
21	0\\
22	0\\
23	0\\
24	0\\
25	0\\
26	0\\
27	0\\
};
\addlegendentry{\scriptsize ML};

\end{axis}
\end{tikzpicture}%